\documentclass[preprint,12pt]{elsarticle}

\usepackage{graphicx}          % Include this line if your 
\usepackage{amsmath,amssymb,amsthm}
\theoremstyle{plain}
\newtheorem{theorem}{Theorem}

\newtheorem{lemma}[theorem]{Lemma}

\theoremstyle{definition}
\newtheorem{definition}[theorem]{Definition}

\newtheorem{example}[theorem]{Example}

\theoremstyle{remark}
\newtheorem{remark}[theorem]{Remark}

\usepackage{subcaption}
\usepackage{siunitx}
\usepackage{accents}
\let\Pr\undefined\DeclareMathOperator{\Pr}{Pr}

\DeclareMathOperator*{\col}{col}

\usepackage{mathtools}

\begin{document}

\begin{frontmatter}

\title{Second-Order Moment-Closure for Tighter Epidemic Thresholds}

\author[naist]{Masaki Ogura\corref{cor}}
\ead{oguram@is.naist.jp}
\fntext[naist]{Graduate School of Information Science, Nara Institute of Science and Technology, Ikoma, Nara 630-0192, Japan}
\cortext[cor]{Corresponding author}

\author[penn]{Victor M. Preciado}
\ead{preciado@seas.upenn.edu}
\fntext[penn]{Department of Electrical and Systems Engineering, University of Pennsylvania, Philadelphia, PA 19104-6314, USA}

\begin{abstract}                          % Abstract of not more than 200 words.
In this paper, we study the dynamics of contagious spreading processes taking
place in complex contact networks. We specifically present a lower-bound on the
decay rate of the number of nodes infected by a susceptible-infected-susceptible
(SIS) stochastic spreading process. A precise quantification of this decay rate
is crucial for designing efficient strategies to contain epidemic outbreaks.
However, existing lower-bounds on the decay rate based on first-order mean-field
approximations are often accompanied by a large error resulting in inefficient
containment strategies. To overcome this deficiency, we derive a lower-bound
based on a second-order moment-closure of the stochastic SIS processes. The
proposed second-order bound is theoretically guaranteed to be tighter than
existing first-order bounds. We also present various numerical simulations to
illustrate how our lower-bound drastically improves the performance of existing
first-order lower-bounds in practical scenarios, resulting in more efficient
strategies for epidemic containment.
\end{abstract}

                             % keyword list or with the 

\end{frontmatter}

\section{Introduction}

Understanding the dynamics of spreading processes taking place in complex
networks is one of the central questions in the field of network science, with
applications in information propagation in social networks~\cite{Lerman2010},
epidemiology~\cite{Nowzari2015a}, and cyber-security~\cite{Roy2012}. Among
various quantities characterizing the asymptotic behaviors of spreading
processes, the \emph{decay rate} (see, e.g.,
\cite{VanMieghem2009a,Chakrabarti2008}) of the spreading size (i.e., the number
of nodes affected by the spread) is of fundamental importance. Besides
quantifying the impact of contagious spreading processes over
networks~\cite{Lajmanovich1976,Ganesh2005}, the decay rate has been used to
measure the performance of containment strategies to control epidemic
outbreaks~\cite{Wan2008IET}. In this direction, the authors
in~\cite{Preciado2014} presented an optimization-based approach for distributing
a limited amount of resources to efficiently contain spreading processes by
maximizing their decay rate towards the disease-free equilibrium. This framework
was later extended to the cases where the underlying network in which the
spreading process is taking place is uncertain~\cite{Han2015a}, temporal
\cite{Ogura2015a,Nowzari2015}, and adaptively changing
\cite{Ogura2015i,Ogura2016l}. Recently, the authors in~\cite{AbadTorres2016}
presented an approach for achieving an optimal resource allocation in order to
maximize the decay rate under sparsity constraints.

However, finding the decay rate of a spreading process is, in general, a
computationally hard problem. Even for the case of the
susceptible-infected-susceptible (SIS) model~\cite{Nowzari2015a}, which is one
of the simplest models of spread, the exact decay rate is given in terms of the
eigenvalues of a matrix whose size grows \emph{exponentially} fast with respect
to the number of nodes in the networks~\cite{VanMieghem2009a}. In order to avoid
this computational difficulty, it is common in the
literature~\cite{Preciado2014,Han2015a,AbadTorres2016} to use a lower-bound on
the decay rate based on first-order mean-field approximations of the spreading
processes. However, this first-order approximation is not necessarily accurate;
in other words, its approximation error can be significantly large for several
important social and biological networks, as we will demonstrate later in this
paper. Therefore, the design of strategies for epidemic containment based on
mean-field approximations can result in inefficient control policies.

The aim of this paper is to present a tighter lower-bound on the decay rate of
the stochastic SIS process based on a second-order moment closure. Specifically,
we show that the decay rate is bounded from below by the maximum real eigenvalue
of a Metzler matrix whose size grows quadratically with respect to the number of
nodes in the network. In order to derive our lower-bound, we describe the
stochastic dynamics of the SIS process using a system of stochastic differential
equations with Poisson jumps. This approach allows us to conveniently evaluate
the dynamics of the first and the \emph{second}-order moments of random
variables relevant for the spreading processes. Furthermore, we prove
theoretically and illustrate numerically that our lower-bound strictly improves
the one based on first-order approximations.

We \label{ref:contributions} remark that, although improved decay rates for the
discrete-time SIS model were presented using second-order analysis in
\cite{Ruhi2016a}, their bounds are applicable only to the special case where the
transmission and recovery rates of nodes are homogeneous and, furthermore,
satisfy restrictive algebraic conditions in terms of nonnegativity of infinitely
many matrices. Likewise, the second-order analysis of the continuous-time SIS
model by the authors in \cite{Cator2012} uses mean-field approximations and,
hence, it is not clear how the analysis relates to the dynamics of the original
stochastic SIS process. {Moreover, their analysis is valid only
when a dominant eigenvalue of a certain matrix (i.e., an eigenvalue having the
maximum real part) is real. In contrast with these limitations of the results in
the literature, our framework applies to the heterogeneous SIS model without any
restrictions, and is supported by rigorous proofs instead of approximations.}

This paper is organized as follows. In Section~\ref{sec:problemStatment}, we
state the problem studied in this paper. In Section~\ref{sec:main}, we present
our lower-bound on the decay rate, and show that this bound strictly improves
the one based on first-order approximations. The effectiveness of our
lower-bound is numerically illustrated in Section~\ref{sec:simulations}.

\subsection{Mathematical preliminaries}

We denote the identity and the zero matrices by $I$ and~$O$, respectively. For a
vector $u$, we denote by $u_{\backslash\{i\}}$ the vector that is obtained after
removing the $i$th element from~$u$. Likewise, for a matrix~$A$, we let $A_{i,
\backslash\{j\}}$ denote the row vector that is obtained after removing the
$j$th element from the $i$th row of~$A$. We say that a square matrix~$A$ is
irreducible if no similarity transformation by a permutation matrix transforms
$A$ into a block upper-triangular matrix. The block-diagonal matrix containing
matrices $A_1$, $\dotsc$,~$A_n$ as its diagonal blocks is denoted
by~$\bigoplus_{i=1}^n A_i$. If the matrices~$A_1$, $\dotsc$, $A_n$ have the same
number of columns, then the matrix obtained by stacking $A_1$, $\dotsc$, $A_n$
in vertical is denoted by $\col_{1\leq i\leq n} A_i$.

A directed graph is defined as the pair $\mathcal G = (\mathcal V, \mathcal E)$,
where $\mathcal V$ is a finite ordered set of nodes and $\mathcal E \subset
\mathcal V \times \mathcal V$ is a set of directed edges. By convention, if $(v,
v') \in \mathcal E$, we understand that there is an edge from $v$ pointing
towards~$v'$, in which case $v$ is said to be an in-neighbor of~$v'$. A directed
path from $v$ to $v'$ in~$\mathcal G$ is an ordered set of nodes $(v_0, \dotsc,
v_{\ell})$ such that $v_{0} = v$, $v_{\ell} = v'$, and $(v_{{k}}, v_{k+1}) \in
\mathcal E$ for $k = 0, \dotsc, \ell-1$. We say that $\mathcal G$ is strongly
connected if there exists a directed path from~$v$ to~$v'$ for all $v, v'\in
\mathcal V$. The adjacency matrix of~$\mathcal{G}$ is defined as the square
matrix, having the same dimension as the number of the nodes, such that its $(i,
j)$th entry equals $1$ if the $j$th node is an in-neighbor of the $i$th node,
and equals $0$ otherwise. It is well known that a directed graph is strongly
connected if and only if its adjacency matrix is irreducible.

A real matrix~$A$ (or a vector as its special case) is said to be nonnegative,
denoted by $A\geq 0$, if all the entries of~$A$ are nonnegative. Likewise, if
all the entries of~$A$ are positive, then $A$ is said to be positive. For
another matrix~$B$ having the same dimensions as $A$, the notation $A\leq B$
implies $B-A\geq 0$. If $A\leq B$ and $A\neq B$, we write $A\lneq B$. For a
square matrix~$A$, we say that $A$ is Metzler~\cite{Farina2000} if the
off-diagonal entries of~$A$ are nonnegative. It is easy to see that $e^{At} \geq
0$ if $A$ is Metzler and $t\geq 0$ (see, e.g., \cite{Farina2000}). For a Metzler
matrix~$A$, the maximum real part of the eigenvalues of~$A$ is denoted
by~$\lambda_{\max}(A)$. In this paper, we use the following basic properties of
Metzler matrices:

\begin{lemma}\label{lem:PF}
The following statements hold for a Metzler matrix $A$:
\begin{enumerate}
\item \label{lem:PF:PF} $\lambda_{\max}(A)$ is an eigenvalue of $A$. Moreover,
if $A$ is irreducible, then there exists a positive eigenvector corresponding to
the eigenvalue $\lambda_{\max}(A)$.

\item \label{lem:PF:monotone} If $A \leq B$, then
$\lambda_{\max}(A) \leq \lambda_{\max}(B)$. Furthermore, if $A$ is irreducible
and $A\neq B$, then $\lambda_{\max}(A) < \lambda_{\max}(B)$.

\item \label{lem:PF:monotone:cor} Assume that $A$ is irreducible. If there exist
a positive vector $u$ and a positive constant $\rho$ such that $Au \lneq \rho
u$, then $\lambda_{\max}(A) < \rho$.
\end{enumerate}
\end{lemma}

\begin{proof}
The first claim is part of the Perron-Frobenius theorem for Metzler matrices
(see, e.g., \cite[Theorems~11 and~17]{Farina2000}). The second claim follows
from the Perron-Frobenius theory and the monotonicity of the maximum real
eigenvalue of nonnegative matrices~\cite[Section~8.4]{Horn1990}. To prove the
last statement, let $\epsilon = \rho u - Au$ and define $A' = A +
\bigoplus(\epsilon_1/u_1, \dotsc, \epsilon_n/u_n)$, where $n$ is the length of
the vector $u$. Since $A'u = A u + \epsilon = \rho u$, $A'$ is irreducible, and
$v$ is positive, it follows that $\lambda_{\max}(A') = \rho$ from the
Perron-Frobenius theorem for irreducible Metzler
matrices~\cite[Theorem~17]{Farina2000}. Since $A$ is irreducible and $A\lneq
A'$, the second statement of the lemma shows that $\lambda_{\max}(A) <
\lambda_{\max}(A')= \rho$.
\end{proof}

\section{Problem Statement}\label{sec:problemStatment}

We start by giving a brief overview of the SIS model~\cite{Nowzari2015a}. Let
$\mathcal G = (\mathcal V, \mathcal E)$ be a strongly connected directed graph
with nodes $v_1$, $\dotsc$, $v_n$. In the SIS model, at a given (continuous)
time~$t \geq 0$, each node can be in one of two possible states: {\it
susceptible} or {\it infected}. When a node~$v_i$ is infected, it can randomly
transition to the susceptible state with an instantaneous rate~$\delta_i > 0$,
called the {\it recovery rate} of node~$v_i$. On the other hand, if an
in-neighbor of node~$v_i$ is in the infected state, then the in-neighbor can
infect node~$v_i$ with an instantaneous rate~$\beta_i$, where $\beta_i > 0$ is
called the {\it infection rate} of node $v_i$. It is easy to see that the SIS
model is a continuous-time Markov process and has a unique absorbing state at
which all the nodes are susceptible. Since this absorbing state is reachable
from any other state, the SIS model reaches this infection-free absorbing state
in a finite time with probability one. The aim of this paper is to study the
stability of this infection-free absorbing state, defined as follows:

\begin{definition}\label{defn:}
Let $\epsilon>0$ and define the probability 
$$
p_i(t) = \Pr(\text{$v_i$ is infected at time $t$}).
$$
We say that the SIS model is
$\epsilon$\nobreakdash-\emph{exponentially mean stable} if there exists a
constant $C>0$ such that, for all nodes $v_i$ and $t\geq 0$, we have $p_i(t)
\leq C e^{-\epsilon t}$ for any set of initially infected nodes at time $t=0$.
Then, we define the \emph{decay rate} of the SIS model as $\rho = \sup\{
\epsilon \colon \text{SIS model is $\epsilon$-exponentially stable} \}$.
\end{definition}

{The notion of the decay rate was studied in, e.g.,
\cite{VanMieghem2009a} and \cite{Chakrabarti2008} for the cases of continuous-
and discrete-time problem settings, respectively, and is closely related to
other important quantities on spreading processes such as epidemic
thresholds~\cite{VanMieghem2009a} and
mean-time-to-absorption~\cite{Ganesh2005}.} Specifically, a basic argument from
the theory of Markov processes shows that the SIS model is
$\epsilon$\nobreakdash-exponentially mean stable for a sufficiently small
$\epsilon>0$ (with a possibly large $C$) and, therefore, it always has a
positive decay rate. However, exact computation of the decay rate is hard in
practice. Even in the homogeneous case, where all nodes share the same infection
and recovery rates, the decay rate equals the modulus of the largest real-part
of the non-zero eigenvalues of a $2^n\times 2^n$ matrix representing the
infinitesimal generator of the SIS model~\cite{VanMieghem2009a}. An alternative
approach for analyzing the decay rate is via upper bounds on the dynamics of the
SIS model based on first-order mean-field approximations. An example of such a
first-order upper bound is described below. Let us define the vector $p(t) =
\col_{1\leq i\leq n}p_i(t)$ containing the infection probabilities of the nodes.
Also, let $A$ be the adjacency matrix of $\mathcal G$ and define the diagonal
matrices~$B = \bigoplus(\beta_1, \dotsc, \beta_n)$ and~$D = \bigoplus(\delta_1,
\dotsc, \delta_n)$. Then, we can show \cite{Preciado2014} the inequality~$p(t)
\leq e^{(BA-D)t}p(0)$, which gives the following lower-bound on the decay rate:
\begin{equation}\label{eq:rho1}
\rho \geq \rho_1 = -\lambda_{\max}(BA-D). 
\end{equation}
Although this lower-bound is computationally efficient to find, there are
several cases in which we can observe a large gap between this bound and the
exact decay rate, as illustrated in the following example:

\begin{example}\label{ex:}
Let us consider the SIS model over a romantic and sexual network in a high
school (\emph{Jefferson}, $n=288$) taken from \cite{Bearman2004}. For
simplicity, we assume a homogeneous transmission rate $\beta_i =
0.9/\lambda_{\max}(A)$ and a (normalized) recovery rate $\delta_i = 1$ for all
nodes. In order to approximately compute the true decay rate~$\rho$, we generate
\num[group-separator={,}]{10000} sample paths of the SIS model over the time
interval $[0, 100]$ starting from the initial state at which all nodes are
infected. From this computation, we obtain a decay rate of $\rho \approx 0.454$.
On the other hand, the first-order approximation in \eqref{eq:rho1} equals
$\rho_1 = 0.1$, whose relative error from $\rho$ is around 78\%.
\end{example}

\begin{remark}\label{rmk:positiveGap}
{We can in fact show that the strict inequality $\rho > \rho_1$
holds in~\eqref{eq:rho1}. Let $v_i$ be a node having the minimum recovery rate
$\delta_{\min} = \min_{1\leq i\leq n}\delta_i > 0$ among all nodes, and consider
the situation where only the node $v_i$ is initially infected. Since $p_i(t)\geq
e^{-\delta_i t}$ for every $t\geq 0$, we have
\begin{equation}\label{eq:strict1}
\rho \geq \delta_i = \delta_{\min}.
\end{equation}
Then, let us take an arbitrary positive constant $b < \min_{1\leq i\leq
n}\beta_i$. Since $\delta_{\min} = -\lambda_{\max}(-D)$ and $-D \leq b A -D$,
Lemma~\ref{lem:PF}.\ref{lem:PF:monotone} shows
\begin{equation}\label{eq:strict2}
-\delta_{\min} \leq
\lambda_{\max}(b A - D). 
\end{equation}
On the other hand, Since $b A - D$ is irreducible (by
the strong connectivity of $\mathcal G$) and $bA-D \lneq BA-D$, we also have
\begin{equation}\label{eq:strict3}
\lambda_{\max}(bA - D) < \lambda_{\max}(BA-D) = -\rho_1
\end{equation}
by Lemma~\ref{lem:PF}.\ref{lem:PF:monotone}. Inequalities
\eqref{eq:strict1}--\eqref{eq:strict3} prove the strict inequality~$\rho >
\rho_1$, as desired. }
\end{remark}

\section{Main Result} \label{sec:main}

As we have demonstrated in Example~\ref{ex:}, there may be a large gap between
the true decay rate $\rho$ and its first-order approximation~$\rho_1$ for the
SIS model. The aim of this paper is to fill in this gap by providing a better
lower-bound on the decay rate. Specifically, the following theorem presents an
improved lower-bound on the decay rate and is the main result of this paper:

\begin{theorem} 
Define the $n^2\times n^2$ Metzler matrix 
\begin{equation*}
\mathcal A \!=\! \begin{bmatrix}
-D & \bigoplus_{i=1}^n (\beta_i A_{i, \backslash\{i\}})
\\
\col_{1\leq i\leq n}(\delta_i V_i) & 
\ \ \bigoplus_{i=1}^n\left(
- \Gamma_i + \col_{j\neq i} \beta_j A_{j, \backslash\{i\}}
\right)
\end{bmatrix}, 
\end{equation*}
where $V_i\in \mathbb{R}^{(n-1)\times n}$ is the matrix satisfying $V_i x =
x_{\backslash\{i\}}$ for all $x\in\mathbb{R}^n$, and $\Gamma_i  =
\bigoplus_{j\neq i}\gamma_{ij}$ for $\gamma_{ij} = \delta_i + \delta_j +
a_{ij}\beta_i$. Define $\rho_2 = -\lambda_{\max}(\mathcal A)$. Then, we have
\begin{equation*}%\label{eq:main}
\rho \geq \rho_2 > \rho_1.
\end{equation*}
\end{theorem}

In order to prove this theorem, we use a representation of the SIS model as a
system of stochastic differential equations with Poisson jumps (see, e.g.,
\cite{Ogura2015i}). For this purpose, we define the variable $x_i(t)$ as
$x_i(t)=1$ if node~$v_i$ is infected at time~$t$, and $x_i(t)=0$ otherwise.
Then, we can see that these variables obey the following stochastic differential
equations with Poisson jumps:
\begin{equation}\label{eq:SDE:firstOrder}
dx_i = -x_i dN_{\delta_i} + \sum_{k=1}^n a_{ik} (1-x_i)x_k dN_{\beta_i},
\end{equation}
where $a_{ik}$ is the $(i, k)$th entry of the adjacency matrix and
$N_{\delta_i}$ and $N_{\beta_i}$ denote stochastically independent Poisson
counters \cite[Chapter~1]{Hanson2007} with rates~$\delta_i$ and $\beta_i$,
respectively. The rest of this section is devoted to the proof of the theorem.
We divide the proof into two parts, namely, the proof of $\rho\geq \rho_2$
(Subsection~\ref{subsec:1}) and $\rho_2>\rho_1$ (Subsection~\ref{subsec:2}).

\subsection{Proof of $\rho\geq \rho_2$}\label{subsec:1}

From the system~\eqref{eq:SDE:firstOrder} of stochastic differential equations,
we can easily show \cite{Ogura2015i} that the expectation $p_i = E[x_i]$ obeys
the differential equation
\begin{align}
\frac{dp_i}{dt} 
&= -\delta_i E[x_i] + \sum_{k=1}^n a_{ik} E[(1-x_i)x_k]\beta_i
\notag 
\\
&
= -\delta_i p_i + \beta_i\sum_{k\neq i} a_{ik} q_{ik}, \label{eq:dot p_i}
\end{align}
where $q_{ij}(t) = E[(1-x_i(t))x_j(t)]$ is a second-order variable representing
the probability that $v_i$ is susceptible and $v_j$ is infected at time~$t$.
Since $q_{ii} \equiv 0$ by definition, in the sequel we do not consider the
variable of the form $q_{ii}$. We next derive differential equations to
characterize the second-order variables~$q_{ij}$. Applying Ito's formula (see,
e.g., \cite[Chapter~4]{Hanson2007}) for stochastic differential equations with
Poisson jumps to the variable $(1-x_i)x_j$, we can show
\begin{equation}\label{eq:SDE:secondOrder}
\begin{aligned}
d((1-x_i)x_j) 
&=
x_ix_j dN_{\delta_j} - (1-x_i)x_j \sum_{k=1}^n a_{ik}x_k dN_{\beta_i}
\\
&\quad - (1-x_i)x_j dN_{\delta_j} + (1-x_i)(1-x_j)\sum_{k=1}^n a_{jk}x_k dN_{\beta_j}
\end{aligned}
\end{equation}
for any distinct pair $(v_i, v_j)$ of nodes. To proceed, we define the
probabilities~$p_{ij}(t) = E[x_i(t) x_j(t)]$ and~$p_{ijk}(t) = E[x_i(t) x_j(t)
x_k(t)]$ for nodes $v_i$, $v_j$, and $v_k$. Then, from
\eqref{eq:SDE:secondOrder}, we can compute the derivative of $q_{ij}$ as
\begin{equation*}
\begin{aligned}
\frac{dq_{ij}}{dt} 
&= 
\delta_j E[x_ix_j] - \beta_i \sum_{k=1}^n a_{ik}E[(1-x_i)x_j x_k] 
\\
&\quad- \delta_j E[(1\!-\!x_i)x_j] + \beta_j \sum_{k=1}^n a_{jk} E[(1\!-\!x_i)(1\!-\!x_j)x_k]\\
&=
\delta_i p_{ij} - \beta_i a_{ij}(p_{j}-p_{ij}) 
%\\
%&\quad
- \delta_j q_{ij} + \beta_j \sum_{k=1}^n a_{jk} (p_k - p_{ik})
- \epsilon_{ij}, 
\end{aligned}
\end{equation*}
where the function $\epsilon_{ij} = \beta_i \sum_{k\neq j} a_{ik}(p_{jk} -
p_{ijk}) + \beta_j \sum_{k=1}^n a_{jk} (p_{jk}-p_{ijk})$ is nonnegative because
$p_{jk}\geq p_{ijk}$ for all nodes $v_i$, $v_j$, and $v_k$. Using the identity
$p_j - p_{ij} = q_{ij}$ and defining the variable~$\gamma_{ij} = \delta_i +
\delta_j + a_{ij}\beta_i$, we obtain
\begin{equation}\label{eq:dot q_ij}
\frac{dq_{ij}}{dt} 
=
-\gamma_{ij} q_{ij} + \delta_i p_j + \beta_j \sum_{k\neq i} a_{jk}q_{ik} - \epsilon_{ij}.
\end{equation}

In order to upper-bound the infection probabilities of the nodes, we define the
vector variables $q_i = \col_{j\neq i} q_{ij}$ and $q = \col_{1\leq i\leq n}q_i$
having dimensions $n-1$ and $n(n-1)$, respectively. Then, we can rewrite the
differential equation~\eqref{eq:dot p_i} as ${dp_i}/{dt} = -\delta_i p_i +
\beta_i A_{i, \backslash\{i\}} q_i$. Stacking this differential equation in
vertical with respect to $i$, we obtain
\begin{equation}\label{eq:dotp}
\frac{dp}{dt} 
= 
-D p + \biggl(\bigoplus_{1=1}^n\beta_iA_{i, \backslash\{i\}}\biggr)q, 
\end{equation}
where $D = \bigoplus(\delta_1, \dotsc, \delta_n)$. Also, from \eqref{eq:dot
q_ij}, it follows that ${dq_{ij}}/{dt} = -\gamma_{ij} q_{ij} + \delta_i p_j +
\beta_j A_{j, \backslash\{i\}}q_i - \epsilon_{ij}$. Stacking this differential
equation with respect to $j \in \{1, \dotsc, n\} \backslash \{i\}$, we see that 
$$
\frac{dq_i}{dt} = - \Gamma_i q_i  + \delta_i V_i p + \left( \col_{j\neq i} (\beta_j
A_{j, \backslash\{i\}})\right)q_i - \col_{j\neq i}\epsilon_{ij}.
$$ 
By stacking this differential equation with respect to $i = 1, \dotsc, n$, we
obtain the following differential equation
$$
\frac{dq}{dt} =
\left(\col_{1\leq i\leq n} \delta_i V_i \right)p + \bigoplus_{i=1}^n \left(- \Gamma_i
+\col_{j\neq i} \beta_j A_{j, \backslash\{i\}} \right) q - \epsilon,
$$ where $\epsilon=
\col_{1\leq i\leq n}\col_{j\neq i} \epsilon_{ij}$ is a vector-valued nonnegative
function. This differential equation and \eqref{eq:dotp} show that the
variable~$r = \col(p, q)$ satisfies $dr/dt = \mathcal A r-\col(0, \epsilon)$.

We are now at the position to prove the inequality~$\rho\geq \rho_2$. Since
$\mathcal A$ is Metzler and $\epsilon(t)$ is entry-wise nonnegative for every
$t\geq 0$, we can obtain the upper bound $ r(t) = e^{\mathcal A t} r(0) -
\int_{0}^t e^{\mathcal A(t-\tau)}\col(0, \epsilon(\tau))\,d\tau \leq e^{\mathcal
A t}r(0)$. This inequality clearly implies that the SIS model is
$\rho_2$\nobreakdash-exponentially mean stable. This completes the proof of the
inequality.

\subsection{Proof of $\rho_2>\rho_1$}\label{subsec:2}

If $\rho_2 \geq \delta_{\min}$, then we trivially have $\rho_2>\rho_1$ because
we know $\delta_{\min} >\rho_1$ from \eqref{eq:strict2} and \eqref{eq:strict3}.
Let us consider the case of $\rho_2 < \delta_{\min}$. Let $u$ be a non-zero
vector of~$\mathcal A$ corresponding to the eigenvalue~$-\rho_2 =
\lambda_{\max}(\mathcal A)$, i.e., assume that $\mathcal A u = -\rho_2 u$. We
split the matrix~$\mathcal A$ as $\mathcal A = \mathcal M - \mathcal N$ using
the matrices
$$
\mathcal M =
\begin{bmatrix}
O & \mathcal M_{12}
\\
O & 
\mathcal M_{22}\end{bmatrix},\ 
\mathcal N = 
\begin{bmatrix}
D & O
\\
\mathcal N_{21} & 
\mathcal N_{22}
\end{bmatrix},
$$
where 
\begin{equation*}
\begin{gathered}
\mathcal M_{12} 
= 
\bigoplus_{i=1}^n \beta_i A_{i,
\backslash\{i\}},
\ 
\mathcal M_{22} 
= 
\bigoplus_{i=1}^n \left( \col_{j\neq i}
(\beta_j A_{j, \backslash\{i\}}) - \bigoplus_{j\neq i} a_{ij}\beta_i \right)
,
\\
\mathcal N_{21} = -\col_{1\leq i\leq n}\delta_i V_i,
\ 
\mathcal N_{22} =
\bigoplus_{i=1}^n \bigoplus_{j\neq i}(\delta_i + \delta_j).
\end{gathered}
\end{equation*}
Then, we have $\mathcal M u = (\mathcal N -\rho_2 I)u$ and, hence, $(\mathcal N
-\rho_2 I)^{-1} \mathcal M u = u$, where the inversion of $\mathcal N -\rho_2 I$
is allowed by our assumption $\rho_2 < \delta_{\min}$. Therefore, the matrix
$(\mathcal N -\rho_2 I)^{-1} \mathcal M$ has an eigenvalue equal to~$1$. Since
this matrix has the following upper-triangular structure
\begin{equation*}
(\mathcal N -\rho_2 I)^{-1} \mathcal M = 
\begin{bmatrix}
O & * \\ O & \mathcal L
\end{bmatrix}
\end{equation*}
for the Metzler matrix $\mathcal L$ defined by 
\begin{equation*}
\mathcal L  = (\mathcal N_{22} -\rho_2 I)^{-1} (
-\mathcal N_{21}(D -\rho_2 I)^{-1} \mathcal M_{12} + \mathcal M_{22} ), 
\end{equation*}
it
follows that $\mathcal L$ has an eigenvalue equal to~$1$. This specifically
implies that 
\begin{equation}\label{eq:lmL>=1}
\lambda_{\max}(\mathcal L) \geq 1. 
\end{equation}

On the other hand, we can obtain an upper bound on~$\lambda_{\max}(\mathcal L)$
as follows. The irreducible matrix $BA-D$ has a positive eigenvector $v$
corresponding to the eigenvalue~$-\rho_1$ by Lemma~\ref{lem:PF}.\ref{lem:PF:PF}.
Define the positive vector~$w = \col_{1\leq i\leq n}v_{\backslash\{i\}}$. Then,
it is easy to see that $\mathcal M_{12} w = BAv$ and $\mathcal M_{22}w =
\col_{1\leq i\leq n}(BAv)_{\backslash\{i\}} - \col_{1\leq i\leq n} \col_{j\neq
i}\beta_j a_{ij}v_i$. Using these equalities and the eigenvalue equation
$(BA-D)v = -\rho_1 v$, we can show that 
\begin{align}
\mathcal L w
&=
(\mathcal N_{22} -\rho_2 I)^{-1}\col_{i} \left(
\delta_i V_i (D -\rho_2 I)^{-1}(D-\rho_1 I) v
\right) \notag
\\
&\quad +(\mathcal N_{22} -\rho_2 I)^{-1}
\col_{i} ((D-\rho_1 I) v)_{\backslash\{i\}} - \phi \notag 
\\
&=
\col_{i}\col_{j\neq i}\left(\frac{\delta_j -\rho_1}{\delta_j -\rho_2}v_j\right) - \phi\notag
\\
&\leq
\left(\max_{1\leq i\leq n} \frac{\delta_i-\rho_1}{\delta_i-\rho_2} \right)w -\phi
\label{eq:thisequation}
\end{align}
for the nonzero and nonnegative vector 
\begin{equation*}
\phi = \col_{1\leq i\leq n}\col_{j\neq
i} \frac{a_{ij}(\beta_i v_j +\beta_j v_i)}{\delta_i+\delta_j-\rho_2}.
\end{equation*}
Since $\mathcal L$ is irreducible (for the proof, see Appendix~A), the
inequality \eqref{eq:thisequation} and
Lemma~\ref{lem:PF}.\ref{lem:PF:monotone:cor} show $\lambda_{\max}(\mathcal L) <
\max_{1\leq i\leq n} (\delta_i-\rho_1)/(\delta_i-\rho_2)$. Since we have already
shown \eqref{eq:lmL>=1}, there must exist an $i$ such that $1<
(\delta_i-\rho_1)/(\delta_i-\rho_2)$. This inequality is equivalent to
$\rho_2>\rho_1$ because both of $\delta_i-\rho_1$ and $\delta_i-\rho_2$ are
positive. This completes the proof.

\section{Numerical Simulations} \label{sec:simulations}

\begin{figure}
\vspace{-3cm}
\begin{minipage}[b]{1\linewidth}
\centering \includegraphics[height=5cm, clip, trim=2cm 0.7cm .5cm 2cm]{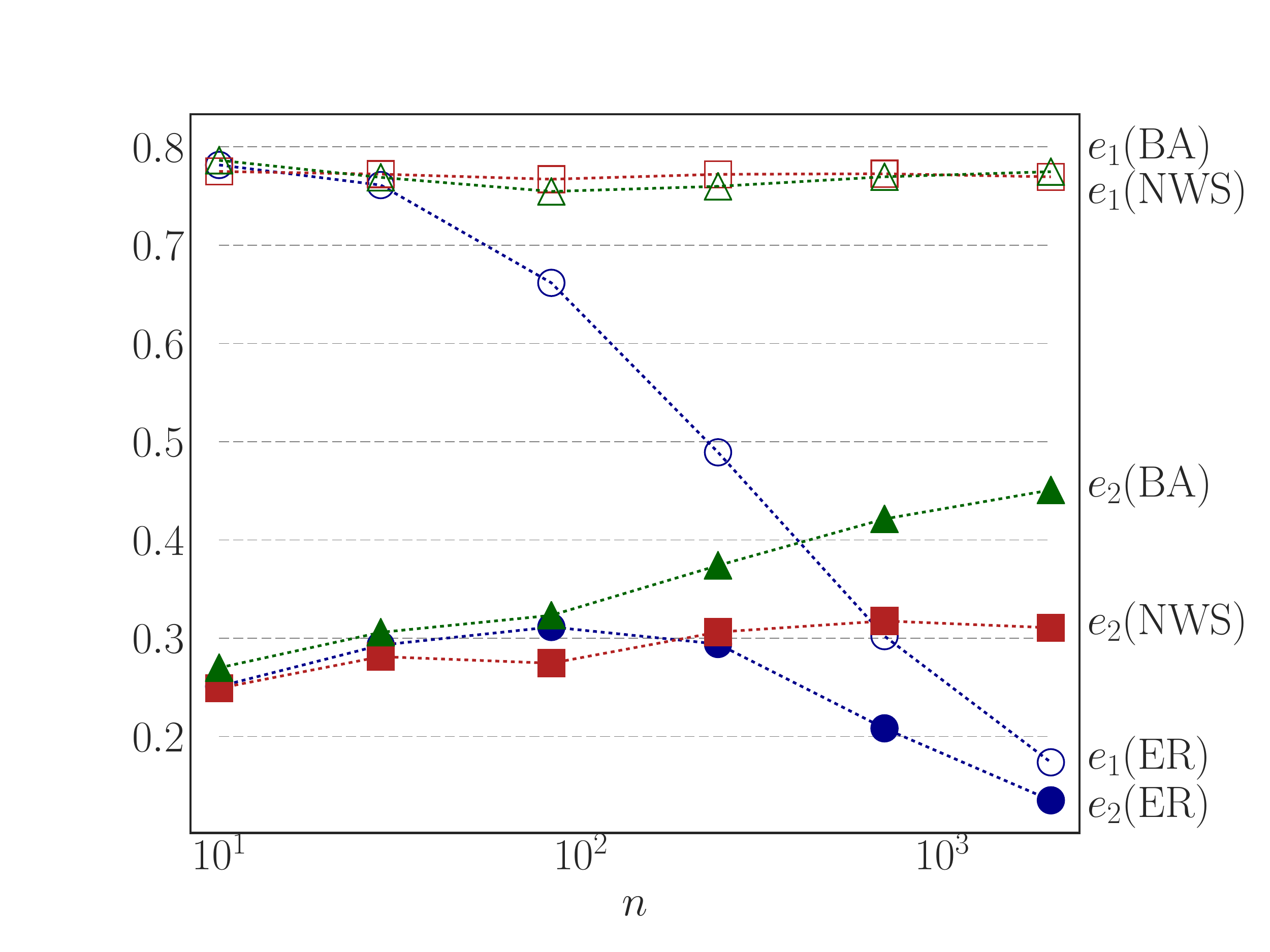}
\subcaption{\footnotesize $\beta =
0.9/\lambda_{\max}(A)$}\label{fig:large}
\vspace{3mm}
\end{minipage}
\begin{minipage}[b]{1\linewidth}
\centering \includegraphics[height=5cm, clip, trim=2cm 0.7cm .5cm 2cm]{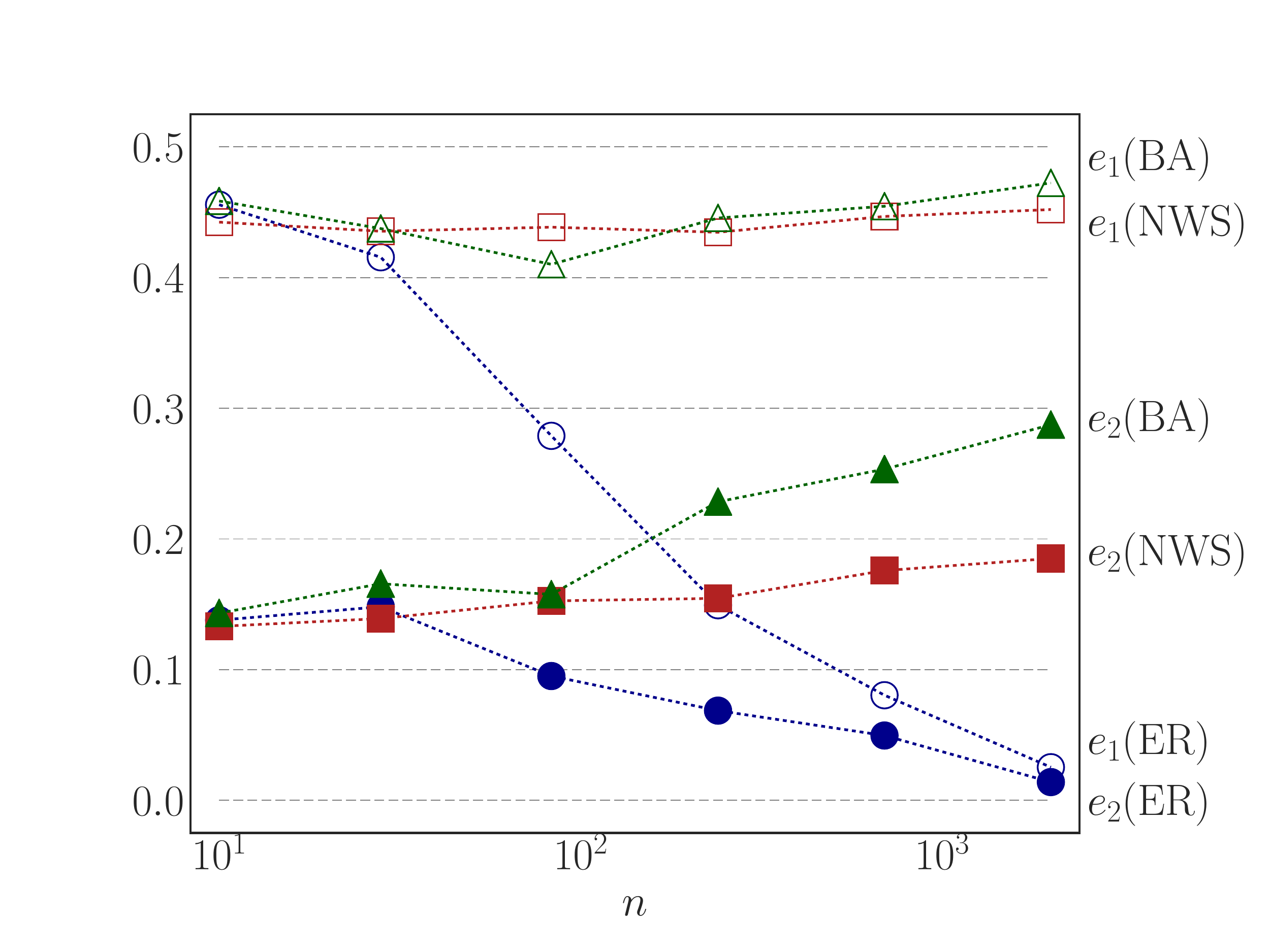}
\subcaption{$\beta =
0.7/\lambda_{\max}(A)$}\label{fig:medium}
\vspace{3mm}
\end{minipage}
\\
\begin{minipage}[b]{1\linewidth}
\centering \includegraphics[height=5cm, clip, trim=2cm 0.7cm .5cm 2cm]{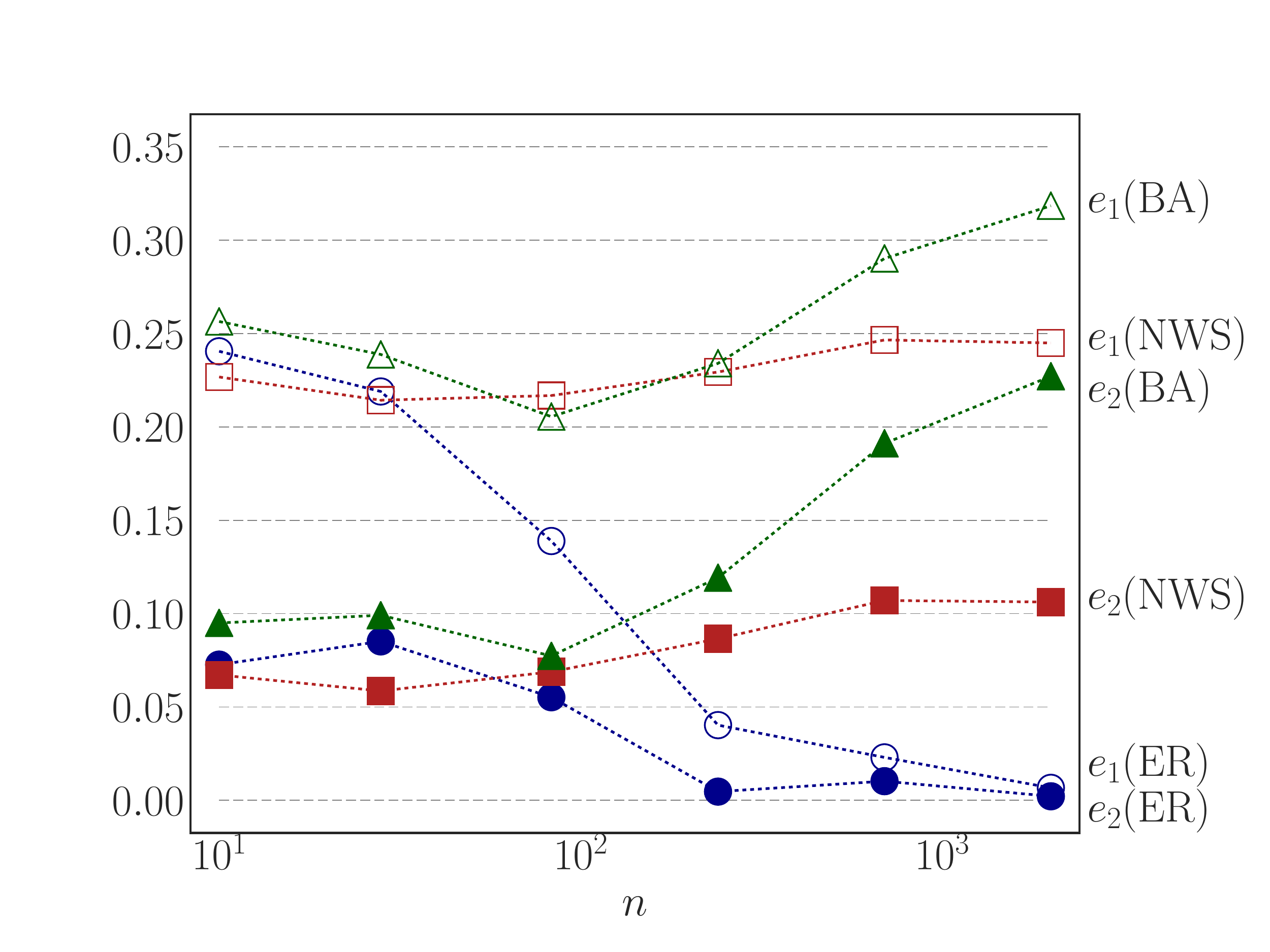}
\subcaption{$\beta =
0.5/\lambda_{\max}(A)$}\label{fig:small}
\vspace{3mm}
\end{minipage}
\caption{Relative approximation errors of the decay rate~$\rho$
for ER (circles), BA (triangles), and NWS (squares) graphs. Empty markers:
first-order bound $\rho_1$, filled markers: second-order bound $\rho_2$.}\label{fig:random}
\end{figure}

In this section, we illustrate the effectiveness of our results with numerical
simulations. The simulations are performed using Python 3.6 on a 4.2 GHz Intel
Core~i7 processor. In our simulations, we consider the SIS model over several
complex networks {with a homogeneous transmission rate~$\beta_i =
\beta$ and a recovery rate $\delta_i = \delta$ for all nodes. We normalize
$\delta = 1$ without loss of generality.} We first consider the following three
random graphs: 1)~Erd\"os\nobreakdash-R\'enyi (ER), 2)~Barab\'asi-Albert (BA),
and 3)~Newman-Watts-Strogatz (NWS) graphs. For each of the networks and various
network sizes, we compute the first-order bound~$\rho_1$, our second-order
bound~$\rho_2$, and an approximation of the true decay rate~$\rho$ (by the same
procedure used in Example~\ref{ex:}). We present the sample averages of the
relative errors $e_1 =(\rho-\rho_1)/\rho$ and $e_2 = (\rho-\rho_2)/\rho$ in
Fig.~\ref{fig:random} (20 realizations of random graphs for each data point)
{for $\beta = 0.9 /\lambda_{\max}(A)$, $0.7/ \lambda_{\max}(A)$,
and $0.5/ \lambda_{\max}(A)$}. {We can observe that our
second-order bound remarkably improves the first-order bound, specifically for
the cases of BA and NWS networks.}

\begin{figure}[tb]
\centering \includegraphics[height=5.4cm, clip, trim=2cm 0.1cm .5cm 2cm]{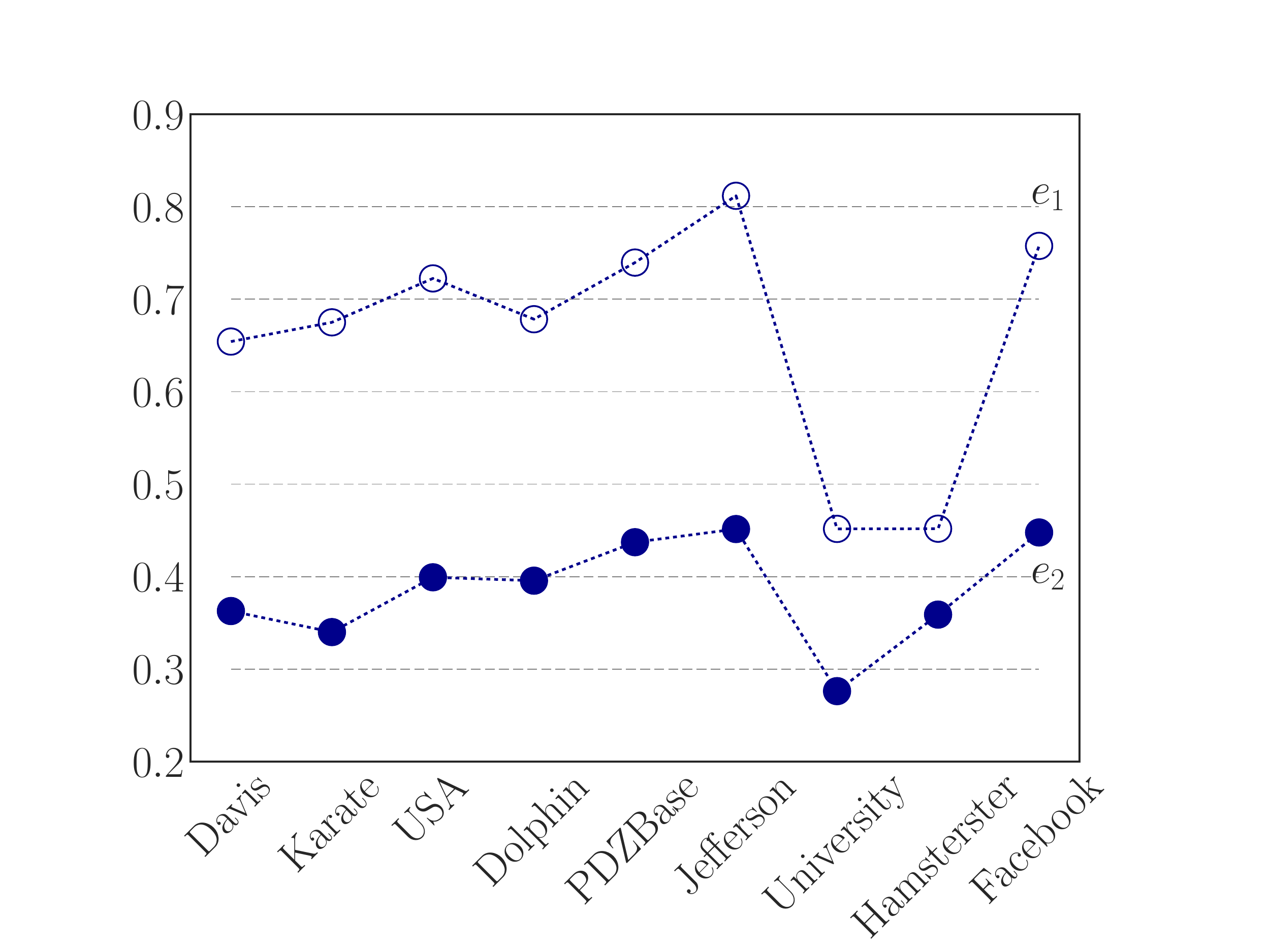}
\caption{Relative approximation errors of the decay rate~$\rho$
for real networks. Empty markers: first-order bound~$\rho_1$,
filled-markers: second-order bound~$\rho_2$.} \label{fig:deterministic}
\end{figure}

We then consider the SIS model over several real-world
networks~\cite{Kunegis2013}. Specifically, we compute lower-bounds on the decay
rates for 1)~a bipartite network from participation of women in social events
(\textit{Davis}, $n=32$), 2)~a social network of the Zachary's Karate club
(\textit{Karate}, $n=34$), 3)~the connectivity network of states in the USA
(\textit{USA}, $n = 49$), 4)~a network of bottlenose dolphins (\textit{Dolphin},
$n=62$), 5)~a network of protein-protein interactions (\textit{PDZBase},
$n=212$), 6) the high-school network described in Example~\ref{ex:}
(\emph{Jefferson}, $n=288$){, 7) an email communication network
at the University Rovira i Virgili (\emph{University}, $n = 1133$), 8) a
friendship network on hamsterster.com (\emph{Hamsterster}, $n = 1858$), and 9) a
friendship network in Facebook (\emph{Facebook}, $n = 2888$)}.
{We consider the homogeneous case as in the above simulations for
random graphs, and use the transmission rate $\beta = 0.9/\lambda_{\max}(A)$ and
the recovery rate $\delta=1$.} We show the relative errors $e_1$ and $e_2$ in
Fig.~\ref{fig:deterministic}. These simulations confirm that our second-order
lower-bound can remarkably improve the first-order bound.

\section{Conclusion}

In this paper, we have presented an improved lower-bound on the decay rate of
the SIS model over complex networks. We have specifically derived a lower-bound
on the decay rate in terms of the maximum real eigenvalue of an $n^2\times n^2$
Metzler matrix, and have shown that our lower-bound improves existing
lower-bound based on mean-field approximations of the SIS model. For deriving
our lower-bound, we have used a linear upper-bounding model for the first and
second-order moments on the SIS model. In our simulations, we have shown that
our lower-bound significantly improves on the first order lower-bound, in the
cases of both random and realistic networks. This improvement suggests that
incorporating second-order moments could allow us to drastically improve the
performance of existing strategies for spreading control
\cite{Wan2008IET,Preciado2014,Han2015a,Ogura2015a,AbadTorres2016}.

\section*{Acknowledgments}

This work was supported in part by the NSF under Grants
CAREER-ECCS-1651433 and IIS-1447470.

\appendix

\section{Irreducibility of $\mathcal L$}

Since $\mathcal N_{22}-\rho_2I$ is a diagonal matrix, it is sufficient to show
the irreducibility of~$\mathcal L' = (\mathcal N_{22}-\rho_2I)\mathcal L$. We
can show that the matrix $\mathcal L'$ can be represented as the block
matrix~$[\mathcal L'_{ij}]_{i, j}$ having the block elements
\begin{equation*}
\mathcal L'_{ij} = \begin{cases}
\col_{j\neq i} \beta_j A_{j, \backslash\{i\}} - \bigoplus_{j\neq i} a_{ij}\beta_i,&\text{if $i=j$, }
\\
\dfrac{\delta_i \beta_j}{\delta_j - \rho_2} V_j(e_j\otimes A_{j, \backslash\{j\}}), &\text{otherwise.}
\end{cases}
\end{equation*}
Therefore, the
irreducibility of~$\mathcal L'$ is equivalent to that of the block
matrix~$\mathcal L'' = [\mathcal L''_{ij}]_{i, j}$ having the block elements
\begin{equation*}
\mathcal L''_{ij} = \begin{cases}
\col_{j\neq i} A_{j, \backslash\{i\}},&\text{if $i=j$, }
\\
V_j(e_j\otimes A_{j, \backslash\{j\}}),&\text{otherwise.}
\end{cases}
\end{equation*}
To prove the irreducibility of~$\mathcal L''$, we notice that $\mathcal L''$
equals the adjacency matrix of the directed graph~$\mathcal G'' = (\mathcal V'',
\mathcal E'')$ having the $n(n-1)$ nodes 
\begin{equation*}
\{v_{1,2}, \dotsc, v_{1,n}, 
v_{2,1}, v_{2,3}, \dotsc, v_{2,n}, 
\dotsc, 
%v_{n-1, 1}, \dotsc, v_{n-1,n-2}, v_{n-1,n}, 
v_{n, 1}, \dotsc, v_{n,n-1}
\}
\end{equation*}
and edges
$\mathcal E'' = \mathcal E''_1 \cup \mathcal E''_2$, where $ \mathcal E''_1 = \{
(v_{i, j}, v_{j, k}) \colon \text{$(j, k) \in \mathcal E$}\} $ and $ \mathcal
E''_2 = \{ (v_{i, j}, v_{i, k}) \colon \text{$(j, k) \in \mathcal E$}\} $. Let
us show that $\mathcal G''$ is strongly connected. Take two arbitrary nodes
$v_{i_0, j_0}$ and $v_{i_1, j_1}$. Since $\mathcal G$ is strongly connected, we
can find a directed-path of the form~$(v_{i_0}, v_{k_1}, \dotsc,
v_{k_{\ell_1-1}}, v_{i_1}, v_{k_{\ell_1}})$ in $\mathcal G$. Then, from the
definition of~$\mathcal E''_1$, we see that the ordered set
$$
( v_{i_0, j_0}, v_{j_0, k_1}, v_{k_1, k_2}, \dotsc,
v_{k_{\ell-2}, k_{\ell-1}}, v_{k_{\ell-1}, i_1}, v_{i_1, k_{\ell}} )
$$ 
is a directed path in $(\mathcal V'', \mathcal E''_1)$. In order to continue
this directed path to $v_{i_1, j_1}$, we take another directed path
$(v_{k_{\ell}}, v_{k_{\ell+1}}, \dotsc, v_{k_{\ell'}}, v_{j_1})$ in $\mathcal
G$. Then, from the definition of~$\mathcal E''_2$, we can see that the ordered
set $ ( v_{i_1, k_{\ell}}, \dotsc, v_{i_1, k_{\ell'}}, v_{i_1, {j_1}} )$ is a
directed path in~$(\mathcal V'', \mathcal E''_2)$. We have thus shown the
existence of a directed path from $v_{i_0, j_0}$ to $v_{i_1, j_1}$ in $\mathcal
G''$. This shows the strong connectivity of~$\mathcal G''$ because $v_{i_0,
j_0}$ and $v_{i_1, j_1}$ were taken arbitrarily. This proves the irreducibility
of~$\mathcal L''$ and, therefore, the irreducibility of~$\mathcal L$, as
desired.

%\bibliographystyle{elsarticle-num} 
%\bibliography{C:/Dropbox/Softwares/TeX/bib/library}

\end{document}